\documentclass[preprint,11pt]{elsarticle}

\usepackage[inner=1in,outer=1in,top=1in,bottom=1in,
marginparwidth=.7in,marginparsep=.1in]{geometry}

\usepackage{mathtools} 
\usepackage{graphicx,multirow}
\usepackage{amssymb,latexsym,amsmath,amsthm,mathrsfs}
\usepackage{color} 
\usepackage[all]{xy} 
\usepackage{caption}
\usepackage{subcaption}
\usepackage{url}

\numberwithin{equation}{section}
\theoremstyle{plain}
\newtheorem{theorem}{Theorem}[section]
\newtheorem{corollary}[theorem]{Corollary}

\newtheorem{conjecture}[theorem]{Conjecture}
\newtheorem{proposition}[theorem]{Proposition}

 \newtheorem*{thm*}{Theorem}

\theoremstyle{definition} 
\newtheorem{remark}[theorem]{Remark} 

\newtheorem{definition}[theorem]{Definition}

 \newtheorem{example}[theorem]{Example}

\usepackage{lipsum}
\usepackage{amsfonts}
\usepackage{graphicx}
\usepackage{epstopdf}
\usepackage{algorithmic}
\ifpdf
\DeclareGraphicsExtensions{.eps,.pdf,.png,.jpg}
\else
\DeclareGraphicsExtensions{.eps}
\fi

\usepackage{enumitem}
\setlist[enumerate]{leftmargin=.5in}
\setlist[itemize]{leftmargin=.5in}

\usepackage{xcolor}
\usepackage[export]{adjustbox}
\usepackage{amssymb,latexsym,amsmath,mathrsfs}

\newcommand\R{\mathbb{R}}
\newcommand\Z{\mathbb{Z}}

\newcommand{\SO}{\mathsf{SO}}
\renewcommand{\O}{\mathsf{O}}
 
\DeclareMathOperator{\Sym}{Sym}
\newcommand\V{\mathcal{V}}

\renewcommand\t{\tau}

\newcommand{\rev}[1]{{\color{black}{#1}}}

\begin{document}

\begin{frontmatter}



\title{The beltway problem over orthogonal groups}


\author{Tamir Bendory,  Dan Edidin, and Oscar Mickelin}


\begin{abstract}
The classical beltway problem entails recovering a set of points from their unordered pairwise distances on the circle. 
This problem can be viewed as a special case of the crystallographic phase retrieval problem of recovering a sparse signal from its periodic autocorrelation.
Based on this interpretation, and motivated by cryo-electron microscopy, we suggest a natural generalization to orthogonal groups: recovering a sparse signal, up to an orthogonal transformation, from its autocorrelation over the orthogonal group.
If the support of the signal is collision-free,  we bound the number of solutions to the beltway problem over orthogonal groups,
and prove that this bound is exactly one when the support of
the signal is radially collision-free (i.e., the support points have distinct magnitudes).  
We also prove that if the pairwise products of the signal's weights are distinct,  then the autocorrelation determines the signal uniquely, up to an orthogonal transformation.
We conclude the paper by considering binary signals 
and show that in this case, the collision-free condition need
not be sufficient to determine signals up to orthogonal transformation.  
\end{abstract}







\end{frontmatter}


\section{Introduction}\label{sec:introduction}
The beltway problem consists of recovering $k$ points $t_1, \ldots, t_k$ on the circle from knowing only the unordered set of their pairwise distances, measured along the circle~\cite{skiena1990reconstructing,lemke2003reconstructing}.
Clearly, the translation of any solution to the beltway problem by a global rotation or reflection is also a solution. 
Two solutions are called \emph{equivalent} if they are related by a rotation and reflection. 
Two sets of points with identical difference sets are termed \emph{homometric}.
Generally, a given set of distances can give rise to non-equivalent solutions. Namely, there may exist homometric sets which are not equivalent. 
In~\cite{skiena1990reconstructing}, the maximum possible number $S(k)$ of  non-equivalent and homometric sets, over all sets of $k$ points on the unit circle, was bounded 
by $ \exp{\left(2^\frac{\ln k}{\ln \ln k}\right)} \leq S(k) \leq \frac{1}{2}k^{k-2}$.

\paragraph{Crystallographic phase retrieval} 
The beltway problem originally arose in X-ray crystallography because of its connection to the phase retrieval problem. The (discrete)  phase retrieval problem entails recovering a signal $x\in\R^n$, up to a sign, cyclic shift and reflection, from the magnitudes of its Fourier transform.
It can be readily seen that this problem is equivalent to recovering $x$ from its periodic autocorrelation~\cite{bendory2022algebraic},
\begin{equation}\label{eq:ac}
	A_\ell(x)= \sum_{i=0}^{n-1}x_ix_{(i+\ell)\bmod n}, \quad \ell=0,\ldots,n-1. 
\end{equation}
Note that the autocorrelation~\eqref{eq:ac} is the second moment of $(g\cdot x),$ where $g$ is a uniformly distributed element of the group of circular shifts $\Z_n$; this observation plays a pivotal role in this paper. 

If $x \in \R^n$ is a binary signal, then the periodic
autocorrelation $A(x)$ is determined by the cyclic difference set of the support of~$x$, and two binary signals have the same autocorrelation if and only if their supports have the same cyclic difference sets. Thus, recovering a binary signal from its periodic autocorrelation~\eqref{eq:ac} reduces to the beltway problem on the vertices of the regular $n$-gon inscribed in the circle. This problem is important since it serves as an approximation of the physical model of X-ray crystallography technology, and thus can be thought of as a special case of the \emph{crystallographic phase retrieval problem}~\cite{elser2017complexity,elser2018benchmark}.
The general crystallographic phase retrieval problem is the problem
of recovering a signal $x \in \R^n$, up to a cyclic shift and reflection, from its second moment, with the classical beltway problem corresponding to binary signals.
If the entries of the signal lie in a small alphabet (which in X-ray crystallography corresponds to different types of atoms~\cite{elser2018benchmark}), then recovering the signal from its autocorrelation is equivalent to checking if several partitions of its support have the same difference sets~\cite{bendory2023finite}. 
 
\paragraph{Existing uniqueness results for phase retrieval} 
The phase retrieval problem is typically ill-posed even for binary signals, and only very recently have uniqueness results been obtained. 
A subset of the authors of this paper established a conjecture that a sparse signal with generic entries is uniquely determined by its periodic autocorrelation as long as $k\leq n/2$~\cite{bendory2020toward}.
In particular,~\cite{bendory2020toward} provides a computational test to check the uniqueness for any pair $(k,n)$. However, because of its heavy computational burden, the conjecture was affirmed only up to $n=9$. 
In~\cite{edidin2023generic}, it was shown that if the signal is sparse with respect to a generic basis in $\R^n$, then $k\leq n/2$ guarantees unique recovery for generic signals, and $k\leq n/4$ suffices to determine all signals.
In~\cite{ghosh2022sparse}, it was shown that \emph{symmetric} signals can be recovered from their periodic autocorrelation for $k=O(n/\log^5 n)$. 
The computational aspects of the beltway problem were studied in~\cite{skiena1990reconstructing,lemke2003reconstructing,duxbury2022unassigned,disser2017geometric,huang2021reconstructing}.

\paragraph{Autocorrelations over orthogonal groups}
The autocorrelation function~\eqref{eq:ac} is the second moment of circular shifts of the signal, where the circular shifts are drawn from a uniform distribution.  
Analogously, one can define an autocorrelation over a compact group $G$ as the expectation of all the products of two entries of $g\cdot x$, where $g$ is drawn from a uniform distribution over $G$. 
The beltway problem over orthogonal groups is that of recovering a discretely supported signal from its autocorrelation over an orthogonal group. 
We formalize the problem rigorously in Section~\ref{sec:orthogonal_beltway} and motivate it by single-particle cryo-electron microscopy: a prominent technology in structural biology to elucidate the spatial structure of biological molecules~\cite{bendory2020single}.

\paragraph{Collision-free support} 
In the classical beltway problem, the points $t_1, \ldots, t_k$ are termed  {\em collision free} if each non-zero distance in the difference set of the points appears with multiplicity one.
If we interpret the beltway problem as recovering a sparse signal from its periodic autocorrelation, then a signal is collision-free if there are no repeated differences in the signal's support.
We generalize this definition to the beltway problem over orthogonal groups, see Definition~\ref{def:collision}.
The collision-free hypothesis is essential to our analysis because, otherwise, it is not necessarily possible to determine the size of the support set from the autocorrelation, except in the binary case when all non-zero entries of the signal are one.

\paragraph{Main contributions}  
The contribution of this paper is three-fold. First, Section~\ref{sec:orthogonal_beltway} formalizes the beltway problem over orthogonal groups.
Second, we prove that the $\O(n)$-orbit of a sparse signal can be uniquely recovered from its autocorrelation if the signal satisfies the radially collision-free condition (i.e., the support points are on distinct spheres). If the signal satisfies an extension of the classical collision-free condition,  we bound the number of orbits that correspond to a given autocorrelation.
Section~\ref{sec.mainresult} presents the main results, which are proved in Section~\ref{sec:proof}. 
The third contribution, in Section~\ref{sec:binary}, is a detailed discussion on binary signals whose support is collision free but not radially collision-free. This case is more subtle. We derive conditions where the binary signal cannot be determined from its second moment, and conjecture when it can be determined. We also discuss the relationship with the turnpike problem. 

\section{The beltway problem over orthogonal groups}\label{sec:orthogonal_beltway}

\subsection{Problem formulation} \label{sec.problem}
The formulation of the classical beltway problem as that of recovering a binary signal from its periodic autocorrelation leads to a natural generalization over orthogonal groups. 
Let $x$ be a signal on $\R^n$; that is, a function $\R^n \to \R.$ The second moment of the translated signal $g\cdot x$ with respect to the uniform (Haar) measure over $\O(n)$ is the function
on $(\R^n)^2$
\begin{equation} \label{eq:ac_groups}
	\begin{split}
		m_2(x)(\t_1,\t_2) &= \int_{\O(n)} (g\cdot x)(\t_1)(g\cdot x)(\t_2)dg.
	\end{split}
\end{equation} 
Note that $m_2$ is $\O(n)$-invariant, namely, $m_2(g\cdot x)=m_2(x)$
for any $g \in \O(n)$, so we can view $m_2(x)$ as a function
$(\R^{n})^2/\O(n) \to \R$.
This, in turn, implies that we can only expect to determine 
the $\O(n)$-orbit of $x$ from its second moment.
In this case,  the second moment is typically referred to as the \emph{autocorrelation}. 

\rev{More generally, it is possible to define the auto-correlation of a (compactly supported) distribution, such as a Dirac delta function, which will produce a distribution supported on
$G$-invariant (compact) subsets of $(\R^{n})^2$ (cf. \cite[Section 2.1]{bendory2023autocorrelation}). In detail, given a collection of points $t_1, \ldots, t_k \in \R^n$ and weights
$w_1, \ldots , w_k \in \R$, we define a $k$-sparse signal on $\R^n$
as
\begin{equation}\label{eq:delta}
	x = \sum_{i=1}^k w_i \delta_{t_i},
\end{equation}
where $\delta_t$ is a point mass located at $t\in \R^n$, represented by the Dirac delta function.  
With a slight abuse of notation, we occasionally refer to signals of the form~\eqref{eq:delta}
as \emph{$\delta$-functions}.
There is an action of $\O(n)$ on $\delta$-functions of the form~\eqref{eq:delta}, given by
$g \cdot x = \sum_{i=1}^k w_i \delta_{g\cdot t_i}.$

In this case, we define the autocorrelation $m_2(x)$ as a distribution
on $(\R^{n})^2$, supported on the $\O(n)$-orbits
$\overline{(t_{i}, t_{j})}:= \{(g\cdot t_i, g\cdot t_j) \,|\, g\in G\}$, where $i,j \in \{1,\ldots,k\}$ as follows.
Observe that for $g \in \O(n)$, the product
$(g \cdot \delta_{t_i})(g \cdot \delta_{t_j})$ is the point mass $\delta_{(g \cdot t_i, g \cdot t_j)}$ supported at $(g \cdot t_i, g \cdot t_j)$ in $(\R^n)^2$; i.e., the distribution characterized by the property that $$\int_{(\R^n)^2} f(\tau_1, \tau_2) \delta_{(g\cdot t_1, g \cdot t_2)}(\tau_1,\tau_2) \;d\tau_1 d\tau_2 = f(g\cdot t_i, g\cdot t_j)$$ for any function $f$ on $(\R^n)^2$. 
We then define the 
integral of these point masses over $\O(n)$
$$\mu_{t_i,t_j}(\tau_1,\tau_2) = \int_{\O(n)} (g \cdot \delta_{t_i})(\tau_1) (g\cdot \delta_{t_j})(\tau_2) \;dg$$ to be the distribution with support the orbit $\overline{(t_i,t_j)}$ 
characterized by property that if $f$ is a function on $(\R^n)^2$ then 
$$\int_{(\R^n)^2} f(\tau_1, \tau_2) \mu_{t_i,t_j}(\tau_1,\tau_2) \;d\tau_1 d\tau_2 =
\int_{\O(n)} f(g\cdot t_i, g\cdot t_j)\;dg.$$
The autocorrelation is then the distribution 
\begin{equation} \label{eq.delta_autocorrelation}
m_2(x)(\tau_1,\tau_2) = \sum_{i, j =1}^k w_i w_j \mu_{t_i,t_j}(\tau_1,\tau_2).
\end{equation}
}

We are now ready to define the beltway problem over orthogonal groups.

\begin{definition}
    The beltway problem over an orthogonal group $\O(n)$ is the task of determining the $\O(n)$-orbit of a $\delta$-function of the form~\eqref{eq:delta} in $\R^n$ from its autocorrelation~\eqref{eq.delta_autocorrelation}.
\end{definition}

In the next section, we state our main results. 
We bound the number of orbits with the same second moment, and then derive conditions for there to be a one-to-one correspondence between the second moment and an orbit.
The results are proved in Section~\ref{sec:proof}. 
Before that, we motivate the beltway problem over orthogonal groups by relating it to \rev{the} single-particle cryo-electron microscopy technology.

\subsection{Motivation: Single-particle cryo-electron microscopy}
The goal of single-particle cryo-electron microscopy is to estimate a 3-D function that represents the electrostatic potential of a biological molecule. The measurements 
are tomographic projections of the molecule, each rotated by a random element of the group of 3-D rotations $\SO(3)$~\cite{bendory2020single}. 
While tomographic projection is not a group action, it was shown that the second moment of cryo-EM is equivalent to~\eqref{eq:ac_groups} with $G=\SO(3)$~\cite{kam1980reconstruction,bendory2022sample}. 
Remarkably, it was proven that, in the high noise regime, the minimal number of samples required for accurate estimation (i.e., the sample complexity) is proportional to the lowest order moment that determines the molecular structure~\cite{abbe2018estimation}. 
Moreover, it was shown that generic molecular structures are determined only by the third moment, implying that the sample complexity scales rapidly~\cite{perry2019sample,bendory2017bispectrum,bandeira2023estimation}. Thus, it is essential to identify classes of molecular structures that can be recovered from lower-order moments. In this context, \cite{bendory2022sample,bendory2023autocorrelation} have shown that if a molecular structure can be represented with only a few coefficients under some basis, then it can be recovered from the second moment, and hence fewer samples are necessary for recovery (i.e., an improved sample complexity). 
These findings were recently extended to any semi-algebraic prior, including sparsity and deep generative models~\cite{bendory2023phase,bendory2024transversality}.
Using only the second moment also reduces the computational burden, and thus it also inspired designing a new class of algorithms, based on the method of moments~\cite{kam1980reconstruction,sharon2020method,levin20183d,bendory2023autocorrelation}. 

\section{Main result} \label{sec.mainresult}

Before stating our main result, we introduce two definitions: an extension of the classical collision-free condition from the circle to $\R^n$, and the \emph{radially collision-free} condition.  The condition that a set $S = \{t_1, \ldots , t_k\}$ of points
in $\R^n$ is radially collision-free is the
assertion that all points of $S$ have distinct magnitudes, a hypothesis
which was used to prove a similar result for cryo-EM in~\cite{bendory2023autocorrelation}.
To motivate the extension of the collision-free definition, note that a periodic signal in $\R^n$ can be viewed as a discrete function on $S^1$, 
supported at a subset $S$ of the  $n$-th roots of unity. 
On the circle, this means that if
$(t_i, t_j)$ and $(t_{\ell}, t_{m})$ are two pairs
of points in $S$, then there is no $g\in \O(2)$ such that
$g\cdot \{t_i, t_j\} = \{t_\ell, t_m\}$, unless $\{i,j\} = \{\ell,m\}$.
The definition below extends this idea from $\O(2)$ to $\O(n)$.

\begin{definition} \label{def:collision}
	Let $S =\{t_1, \ldots , t_k\} \subset \R^n$ be a set of points.
	\begin{enumerate}
		\item We say that
		$S$ is {\em collision-free} if for every pair of two-element
		subsets $\{t_i, t_j\}$ and $\{t_\ell, t_m\}$, $g\cdot \{t_i, t_j\} = \{t_\ell, t_m\}$
		for some $g \in \O(n)$ if and only if $\{t_i, t_j\} = \{t_\ell, t_m\}$.
		
		\item We say that $S$ is {\em radially collision-free}
		if for every $t_i, t_j \in S$, $t_i = g\cdot t_j$ for
		some $g \in \O(n)$ if and only if $t_i =t_j$.
	\end{enumerate}
\end{definition}
Clearly, any radially collision-free set is collision-free, but the converse need not be true.
Note also that the radially collision-free condition has no analog in the classical beltway problem, where all points lie on a circle.

Our main result states that for any collision-free $\delta$-function of the form~\eqref{eq:delta}, there are a finite
number of $\O(n)$-orbits of $\delta$-functions with the same
second moment as $x$. Moreover, we provide an explicit upper bound 
for the number of possible $\O(n)$-orbits with the same second moment.
Notably, this bound implies that if the support of $x$ is radially collision-free, 
then the second moment uniquely determines the $\O(n)$-orbit of $x$. 

To state our main result, we introduce some notation. Suppose that $x$ is a $\delta$-function of the form~\eqref{eq:delta}.
Let $\{z_1, \ldots z_q\}$ be the set of distinct magnitudes of the vectors
$t_i$, and suppose that $r_p$ of the vectors have magnitude $z_p$
with $\sum_p r_p = k$. In particular, the support of $x$ is radially collision-free 
if and only if $q = k$, so all $r_p = 1$.

\begin{theorem} \label{thm.secondmoment}
	If $k \geq 3$, the number of $\O(n)$-orbits of $\delta$-functions of the form~\eqref{eq:delta} with the same second moment
	as $x = \sum_{i =1}^k w_i \delta_{t_i}$, whose support is collision-free,
	is at most
	\begin{equation} \label{eq.bound}
		\prod_{r_p \geq 2} \frac{\binom{r_p}{2}!}{r_p!}\prod_{a < b \leq q} (r_a r_b)!
	\end{equation}
	In particular, if the support of $x$ is radially collision-free (i.e., all
	$r_p =1$),  
	then the second moment determines $x$ up to a translation by an element
	of $\O(n)$.
\end{theorem}
Importantly, Theorem~\ref{thm.secondmoment} is independent of the non-zero weights $w_i$. Thus, it holds for the binary case when the weights are all ones---as in the classical beltway problem---as well as when the weights are either generic~\cite{bendory2020toward} or drawn from a finite alphabet~\cite{bendory2023finite}.
When the weights are sufficiently generic, we can derive a stronger result: 
the orbit of any signal with collision-free support (not necessarily radially collision-free)   is determined by the second moment. 

\begin{corollary} \label{cor.distinctweights}
  Suppose that a $\delta$-function of the form~\eqref{eq:delta} has a  collision-free support.
  If the pairwise products of the weights $\{w_i w_j\}_{i < j}$ are all
  distinct, then the second moment determines the $\O(n)$-orbit of $x$.
\end{corollary} 

Note that if the weights $w_i$  take negative values, then the second moment is invariant to sign, namely, $x$ and $-x$ result in the same second moment. Thus, in this case, $x$ is at best determined up to an element of $\O(n)\times \mathbb{Z}_2$.

\section{Proofs} \label{sec:proof}

\subsection{Proof of Theorem \ref{thm.secondmoment}} 

To prove Theorem \ref{thm.secondmoment},  we use the representation-theoretic
analysis of the second moment developed in~\cite{bendory2022sample}.
The set $\V_{n,k}$ of $\delta$-functions supported at 
$k$-points can be identified with the quotient algebraic variety
$U(n,k)/\Sigma_k$,
where $U(n,k)$ is the set of $k$-tuples
$\left((w_1, t_1), \ldots ,(w_k,t_k)\right)$, the $t_k$ are distinct points in $\R^n$, none of the weights $w_i$ are zero, and
$\Sigma_k$ is the symmetric group of 
permutations of $\{1, \ldots , k\}$. 

The quotient map $\pi \colon U_{n,k} \to \V_{n,k}$ 
is given explicitly by the formula
$$\left((w_1, t_1), \ldots ,(w_k,t_k)\right)\mapsto \sum_{i=1}^k w_i \delta_{t_i}.$$
The set $U(n,k)$ is an $\O(n)$-invariant Zariski open set in the representation $V_{n,k} = (\R \times \R^n)^k$ of $\O(n)$, where $g\in O(n)$ acts by the rule
$g \cdot \left((w_1, t_1), \ldots ,(w_k,t_k)\right)
= \left((w_1, g\cdot t_1), \ldots ,(w_k,g \cdot t_k)\right).$
With this action, the quotient map $\pi \colon U_{n,k} \to \V_{n,k}$
is $\O(n)$ equivariant, meaning that 
$\pi(g \cdot \left((w_1, t_1), \ldots ,(w_k,t_k)\right) = g \cdot 
\pi\left((w_1, t_1), \ldots ,(w_k,t_k)\right).$

Our strategy for proving Theorem \ref{thm.secondmoment} is to use
the $\O(n)$-equivariance of the map $\pi$ to compute
the second moment of $\delta = \pi\left((w_1, t_1), \ldots ,(w_k,t_k)\right)$
in terms of the second moment of the vector 
$\left((w_1, t_1), \ldots ,(w_k,t_k)\right) \in V_{n,k}$.
To simplify notation, we denote elements of $V_{n,k}$ by pairs
$(W,X)$, where $W$ is the $1\times k$ matrix of
weights $(w_1, \ldots , w_k)$ and
$X=(t_1 \ldots t_k)$ is the $n \times k$ matrix whose columns are the vectors
$t_1, \ldots , t_k$. 
The vector space $V_{n,k}$ decomposes as a representation of $\O(n)$, 
$V_{n,k} = V_0^k + V_1^k$, where $V_0$ is the trivial
representation (i.e., $\O(n)$ acts trivially) and $V_1$ is 
the {\em defining representation} of $\O(n)$, corresponding
to $\O(n)$ acting as rigid motions of $\R^n$.
Under this decomposition, the matrix
$W$ corresponds to the projection to $V_0^k$,  and the matrix $X$ corresponds to the projection to $V_1^k$.

By~\cite[Theorem 2.3]{bendory2022sample}, the second moment
of a pair $(W,X)$ is equivalent to the pair of symmetric $k \times k$ matrices $(W^T W, X^T X)$.
Thus, we can identify the second moment on $V_{n,k}$ 
as a function $m_2 \colon V_{n,k} \to \Sym^k \R \times \Sym^k \R, (W,X) \mapsto (W^TW, X^TX)$.
By \cite[Theorem 2.3]{bendory2022sample}, the ambiguity
group of the second moment on the representation $V_{n,k} = V_0^k \oplus V_1^k$ is the group $O(1) \times O(n)$,
since $V_0$ and $V_1$ have dimensions $1$ and $n$, respectively. (Note that in our formulation the ambiguity groups
are orthogonal groups rather than unitary groups because we work with real representations.)
In other words, $m_2(W',X') = m_2(W,X)$ if and only if $W' = \pm W$ and $X' = g X $ for some $g \in \O(n)$.

Let $x = \sum_{i=1}^k w_i\delta_{t_i}$ be a $\delta$-function with support
$\{t_1, \ldots , t_k\}$, let $X = (t_1\; \ldots \; t_k)$ (with respect
to any ordering of the support), let $A = X^T X$ be the corresponding
Gram matrix, and let $\mathcal{T}(A)$ be the set of triples $\{(A_{ii}, A_{jj}, A_{ij})\}_{i < j}$.
\begin{proposition} \label{prop.equiv}
  If $x= \sum_{i=1}^k w_i \delta_{t_i} \in \V_{n,k}$
is collision free, then the second moment $m_2(x)(\tau_1,\tau_2)$ determines $\mathcal{T}(A)$ \rev{and the products of the weights $\{w_iw_j\}_{1 \leq i,j \leq k}$}.
 \end{proposition}
\begin{proof}
\rev{As explained in Section~\ref{sec.problem}, the second moment is the distribution on $\R^n \times \R^n$
$$m_2(x)(\tau_1, \tau_2) = \sum_{i,j =1}^k w_iw_j \mu_{t_i,t_j}(\tau_1,\tau_2),$$
where $\mu_{t_i,t_j}(\tau_1,\tau_2)$ is 
supported on the orbit $\overline{(t_i,t_j)}$. If we assume that the points $t_1, \ldots , t_k$ are collision free,
then the set of $\O(n)$-orbits $\{(t_i,t_j)\}_{i \leq j}$
are all distinct, so we can determine
this set of $\O(n)$-orbits from the support of
the distribution $m_2(x)(\tau_1,\tau_2)$ on $(\R^n)^2$,
and we can likewise determine the products of the weights 
$w_iw_j$ by integrating the characteristic function $\mathbf{1}_{\overline{(t_i,t_j)}}$ with respect to $m_2(x)(\tau_1, \tau_2)$.}
On the other hand,
the $\O(n)$-orbit of the pair $(t_i,t_j)$ is uniquely determined
by the triple of real numbers
$(t_i \cdot t_i, t_j\cdot t_j, t_i \cdot t_j)$. Thus, if  $A = X^TX$ is the Gram matrix of $X$, then the orbit
$\overline{(t_i,t_j)}$ determines the triple of entries $(A_{ii}, A_{jj}, A_{ij})$ and the proposition follows.
\end{proof}        
Proposition \ref{prop.equiv} implies that if the support of $x$ is collision-free, then there are a finite number of $\O(n)$-orbits of $\delta$-functions $y$
with the same second moment. This follows since if $y = \sum_{i=1}^k w'_i \delta_{s_i}
\in \V_{n,k}$
is another $\delta$-function with the same second moment
as $x = \sum w_i \delta_{t_i}$, then
$\mathcal{T}(A) = \mathcal{T}(B)$ where $A = X^T X$ and $B = Y^T Y$, where
$Y = (s_1\;\ldots \; s_k)$.
We complete the proof of Theorem~\ref{thm.secondmoment}
by bounding this number. 

Since $A$ and $B$ have the same set of diagonal entries with multiplicity,
we can, after reordering the 
supports, assume that $\|t_i\| =\|s_i\|$ for all $i$
and that $\|t_i\| \leq \|t_{i+1}\|$ for $i = 1, \ldots, k$. By assumption, $\|t_i\|$
takes on $q$ different values as $i$ ranges from 1 to $k$.
Let $\mathcal{P}_1, \ldots ,\mathcal{P}_q$ be the partition of $\{1,\ldots,k\}$
such that $\|t_j\|$ is constant for all $j \in \mathcal{P}_\ell$ and $|\mathcal{P}_\ell| = r_\ell$.

The symmetric matrices
$A = X^TX$ and $B = Y^TY$ have the same set of entries and identical diagonals.
For any $i \leq j$, the triplet $(A_{ii}, A_{jj}, A_{ij})$ represents
\rev{an orbit} in the support of $m_{2}(x)$. Likewise, for any $k \leq l$, the triple 
$(B_{kk}, B_{ll}, B_{kl})$ \rev{represents an orbit} in the support of
$m_2(y)$. Since $x$ and $y$ have the same second moments, the
sets of triples $\{(A_{ii}, A_{jj}, A_{ij})\}_{i,j}$ and
$\{(B_{kk},B_{ll}, B_{kl})\}_{k,l}$ must be the same.
Let $B_{kl}$ be an entry of $B$ with $k \leq l$, and 
with $k \in \mathcal{P}_a, l \in \mathcal P_{b}$. 
Because we have ordered
the diagonal, we know that $a\leq b$ and that  if
	$B_{kl} = A_{ij}$ for some $i \leq j$, then $(i,j) \in \mathcal{P}_a \times \mathcal{P}_b$.
Hence, if $a <  b$, the entry $B_{kl}$ can take $r_a r_b$ possible values corresponding to pairs in the product $\mathcal{P}_a \times \mathcal{P}_b$.
On the other hand, if $a = b$, then $B_{kl}$ can take $\binom{r_a}{2}$ possible
values 
corresponding to pairs $i \leq j \in \mathcal{P}_a \times \mathcal{P}_a$.
Thus, there are at most $\prod_{r_p \geq 2} \binom{r_p}{2}! \prod_{a < b\leq q} (r_n r_m)!$
possible matrices $B = Y^TY$ corresponding to Gram matrices
of $\delta$-functions $y$ with the same second moment as $x$. However,
if we reorder the points $t_i$ with the same magnitude we do not change the diagonal and obtain a matrix $Y$ which is a permutation
of $X$, and the Gram matrix $B = Y^TY$ corresponds to 
a $\delta$-function $y$ with the same
$\O(n)$-orbit as $x$. Since there are $\prod_{p=1}^q r_p!$ reorderings
of the support set $t_1, \ldots, t_k$ which preserve the magnitudes, we see that there are at most $\prod_{r_p \geq 2} \frac{\binom{r_p}{2}!}{r_p!} \prod_{a < b \leq q} (r_a r_b)!$ 
possible $\O(n)$-orbits of $\delta$-functions with the same second moment
as $x$.

\begin{remark}
  Note that the bound given in~\eqref{eq.bound} merely determines
  the number of permutations of $A = X^T X$, which could possibly be
  the Gram matrix of a matrix $Y$, which is not a permutation of
  $X$. However, in any example, the number of possible matrices $Y$ with
  $\mathcal{T}(Y^TY) = \mathcal{T}(X^T X)$ will be further limited by
    the fact that the permutation of $A$ must necessarily be a positive semidefinite matrix of rank $\ell$,  
    where $\ell$ is the dimension of the subspace spanned by the vectors in the support of the function $x$. However, we do not
    know a way to use this constraint to reduce the bound uniformly over all
    collision-free $x \in  \V_{n,k}$; see Example~\ref{ex.piccard}.
\end{remark}   

\subsection{Proof of Corollary \ref{cor.distinctweights}}
If $x = \sum_{i=1}^k w_i \delta_{t_i}$, then by Proposition~\ref{prop.equiv} the autocorrelation $m_2(x)(\tau_1,\tau_2)$ determines
the $\O(n)$-orbits $\overline{(t_i, t_j)}$
and the values of $w_i w_j$. In particular, if all pairwise products
$w_i w_j$ are distinct, then any $\delta$-function $y = \sum w_i' s_i$
with the same second moment as $x$
must have (after possibly reordering the support point $s_i$),
$w'_i = \pm w_i$ and $\overline{t_i, t_j}= \overline{s_i, s_j}$, 
since $\overline{t_i,t_j}$ is the unique orbit where the \rev{integral
with respect to $m_2(x)(\tau_1,\tau_2)$
of its characteristic function
is $w_i w_j$.}
In particular, if $X = (t_1
;\ldots \; t_k)$ and $Y = (s_1\; \ldots \; s_k)$, 
then for any pair $i \leq j$, $(X^TX)_{ij} = (Y^TY)_{ij}$; i.e., $X$ and $Y$
have the same Gram matrices so $x$ and $y$ are orthogonally equivalent.

\section{The beltway problem over orthogonal groups for binary signals}\label{sec:binary}
Theorem \ref{thm.secondmoment} implies that the $\O(n)$-orbit
 of any  $\delta$-function whose support is radially collision-free
 is uniquely determined by its second moment. Likewise, by Corollary~\ref{cor.distinctweights}, if the support of $x$ is only collision-free
 but the pairwise products of the weights are distinct, then 
 the $\O(n)$-orbit of $x$ is determined from the second
 moment. 
In this section, we discuss the problem of determining a \emph{binary} $\delta$-function in $\R^n$ from its autocorrelation over the orthogonal group $\O(n)$, when the support is collision-free but not radially collision-free.  
In this case, the situation is more nuanced, and we divide the discussion into four subsections. 
In Section~\ref{sec:binary_independent} we prove that if the support of $x$ consists
 of linearly independent points, then we cannot expect to recover
 the $\O(n)$-orbit of $x$ from its second moment. 
 In Sections \ref{sec:sphere} and~\ref{sec:sphere_non_independent} we discuss the case where the support set lies on a sphere. In Section~\ref{sec:sphere} we prove that 
if $k\leq n$, then we cannot determine
 the $\O(n)$-orbit of $x$ from its second moment. 
 Conversely, in Section~\ref{sec:sphere_non_independent}, we establish a conjecture that if $k>n$ then an $\O(n)$-orbit of a generic binary  $\delta$-function is determined by the second moment. 
Finally, in Section~\ref{sec:binary_turnpike} we discuss the connection of this problem with the \emph{turnpike} problem. 
 
 \subsection{Linearly independent supports} \label{sec:binary_independent}
In this section, we show that if the support of a binary $\delta$-function
$x$ consists of linearly independent points in $\R^n$, then we cannot expect to recover the $\O(n)$-orbit of $x$ from its second moment if at least two
of the points have the same magnitude.

\begin{proposition} \label{prop.indep}
Consider the set of $\O(n)$-orbits of
  binary  $\delta$-functions $x = \sum_{i}^k \delta_{t_i}$, 
  with $k \leq n$, such
  that $S = \{t_1, \ldots , t_k\}$ is collision-free and $\|t_i\| = \|t_j\|$
  for some $i \neq j$.   Then, there exists a 
  Zariski dense subset $\mathcal{W}$ 
such that for all $x \in \mathcal{W}$, the $\O(n)$-orbit
  of $x$ is not determined by its second moment.
\end{proposition}

\begin{proof}
 Let $\mathcal{U}$ be the set of $\delta$-functions satisfying the hypotheses
 of the proposition. The subset of $\mathcal{U}' \subset \mathcal{U}$ corresponding
 to $\delta$-functions $x$ whose
supports are linearly independent is Zariski open and hence
 dense. Thus, it suffices to prove that there is a Zariski dense
 subset of $\mathcal{U}'$ such that if $x \in \mathcal{U}'$, the $\O(n)$-orbit of $x$ is not determined from its second moment.
 
 Suppose $x = \sum_{i=1}^k \delta_{t_i} \in \mathcal{U'}$.
 After reordering
  the points, we may assume that $\|t_1\| = \|t_2\|$. Also, 
 after applying a suitable
  orthogonal transformation, we may assume that the matrix
  $X = (t_1 \ldots t_k)$ is upper triangular; i.e., $t_i = (t_{i1}, \ldots ,t_{ii}, 0, \ldots ,0)^T$ for $i = 1, \ldots ,k$.
  In particular, we may assume that the last $n-k$ rows of $X$
  are zero. It follows that the $k \times k$ Gram matrix $X^TX$ is the same as the Gram matrix obtained by deleting the last $n-k$ rows
  of the triangular matrix $X$. In other words, we can, without loss of generality, assume that $k = n$.
  
  Now, consider matrices of the form $Y = (t_1\; \ldots t_{k-1}\; s_k)$.
  The Gram matrices $X^TX$ and $Y^TY$ differ only in the $k$-th row and column.
Consider the set of $s_k$, which satisfy the equations $t_1 \cdot s_k  =  t_2 \cdot t_k$, $t_2 \cdot s_k  =  t_1 \cdot t_k$, $t_i \cdot s_k  = t_i \cdot t_k$ for $i=3,\ldots,k-1$, and $s_k \cdot s_k  =  t_k \cdot t_k$.
Because the matrix $X = (t_1 \ldots t_{k-1})$ is triangular,
  if we write $s_k = (s_{k1} \ldots s_{kk})^T$, the first $k-1$ equations
  in the system above give a unique solution for
  $(s_{k1}, \ldots , s_{k,k-1})$. The last coordinate $s_{kk}$ is determined, up to a sign, by the equation $s_k \cdot s_k = t_k \cdot t_k$. However, this equation
  will not have a solution if $\sum_{i=1}^kt_{ki}^2 < \sum_{l=1}^{k-1} s_{kl}^2$. However,
  since $s_{k1}, \ldots s_{k,k-1}$ do not depend on $t_{kk}$, we know that
  for fixed $t_{k1}, \ldots , t_{k,k-1}$ and $t_{kk}$ sufficiently large
  the system will have a solution. Moreover, this solution is unique up to the action of $\O(k)$.
 By construction, if $Y = (t_1\; \ldots t_{k-1}\; s_k)$, 
  $\mathcal{T}(Y^TY) = \mathcal{T}(X^TX)$ but the Gram matrix $Y^TY$ cannot
    be obtained from a permutation of $X$. Hence,
    if $y =\delta_{t_1} + \ldots +\delta_{t_{k-1}} + \delta_{s_k}$, then
    $m_2(y)= m_2(x)$ but $y$ is not orthogonally equivalent to $x$. 
\end{proof}

\subsection{$\delta$-functions supported on spheres with $k\leq n$} \label{sec:sphere}
Let us consider binary $\delta$-functions of the form $\sum_{i=1}^k \delta_{t_i}$,
where $S = \{t_1, \ldots , t_k\}$ is a collision-free subset of points
in $S^{n-1}$. By Theorem~\ref{thm.secondmoment},
we know that there can be up to $\frac{\binom{k}{2}!}{k!}$ possible
non-orthogonally equivalent 
$\delta$-functions  with the same second moment 
as $x$ (including $x$).
When $k \leq n$, then a general
collection of points on $S^{n-1}$ are linearly independent and 
the method used in the proof of Proposition~\ref{prop.indep} yields
the following result that bounds the number of solutions from below. 

\begin{proposition}
Consider the set
of binary $\delta$-functions on the sphere whose support is collision-free and assume $3<k \leq n$. Then, there is a Zariski dense subset $\mathcal{W}$ so that if $x \in \mathcal{W}$, then there are at least
$(k-1)!$ non-equivalent binary $\delta$-functions with the same
second moment as $x$ (including $x$).
\end{proposition}

\begin{example}
We performed the following numerical experiment with 
$k = n =4$.
We constructed  $10,000$
random $4 \times 4$ triangular matrices $X = (t_1\;t_2\;t_3\;t_4)$ with 
each $\|t_i\| = 1$ as follows.
We took $t_1 = (1,0,0,0)^T$, $t_2 = (t_{21},t_{22},0,0)^T$ with
$(t_{21},t_{22})$ uniformly (Haar) sampled on $S^1$,
$t_3  = (t_{31},t_{32}, t_{33},0)^T$ with
$(t_{31},t_{32}, t_{33})^T$  uniformly sampled on $S^2$,
and $t_4 = (t_{41}, t_{42}, t_{43}, t_{44})^T$ uniformly sampled
on $S^3$. We found 
that in approximately $14\%$ of the
sampled matrices, any non-trivial permutation of the first three entries of the 
fourth column of the Gram matrix $X^TX$ yields the Gram
matrix of another triangular matrix $(t_1\; t_2 \;t_3 \;s_4)$
with $\|s_4\|=1$;
this yields
$\delta$-functions $y = \delta_{t_1} + \delta_{t_2} + \delta_{t_3} + \delta_{s_4}$ supported on $S^3$ which are not orthogonally equivalent to $x = \sum_{i=1}^4 \delta_{t_i}$. Note that Theorem~\ref{thm.secondmoment}
implies that there are at most $6!/4! = 30$ possible $\delta$-functions
with the same second moment as $x = \delta_{t_1} + \delta_{t_2} + \delta_{t_3} + \delta_{t_4}$ (including $x$). 
\end{example}

\subsection{{$\delta$-functions supported on spheres with $k>n$}} \label{sec:sphere_non_independent}

When $k > n$, the support of a $\delta$-function consists
of linearly dependent points. In this case, we pose
the following conjecture.
\begin{conjecture} \label{question.generic}
    Suppose that $k > n$. Then, the  $\O(n)$-orbit of a generic binary  $\delta$-function $x = \sum_{i=1}^k\delta_{t_i}$ with support
    on the sphere $S^{n-1}$ is determined by its second moment. 
\end{conjecture}

If $x = \sum_{i=1}^k \delta_{t_i}$ is supported on $S^n$ and 
has a collision-free support,  
then the second moment determines all pairwise inner products 
$t_i \cdot t_j$. Since $\|t_i\|=1$ for all $i$, this is equivalent to knowing the pairwise distance $\|t_i -t_j\|$. Thus, the problem of recovering
the $\O(n)$-orbit of this $\delta$-function is equivalent to recovering the $\O(n)$-orbit of $S = \{t_1, \ldots t_k\}$ from their pairwise distances. 
The problem of determining a set of points in $\R^n$
from their pairwise distances (up to rigid motions) was studied in~\cite{boutin2007point}, where the authors prove that if $k \geq n+2$, then 
there is a hypersurface in $\Sym^k\R^n$ such that any
set of $k$ points in the complement of this hypersurface
can be recovered, up to a rigid motion, from its set of pairwise distances. 
This result can be viewed as giving evidence that Conjecture~\ref{question.generic} has a positive answer, at least when $k \geq n+2$.

\subsection{$\delta$-functions supported on $S^1$ and the turnpike problem} \label{sec:binary_turnpike}

The {\em turnpike problem} is the problem of recovering a set of points
$S = \{a_1, \ldots , a_k\}\subset \R$, 
up to a translation and reflection, from their pairwise distances
$|a_i -a_j|$. 
 This problem arises in multiple applications, including angle-of-arrival estimation~\cite{karanam2018magnitude}, identifiability of quantum systems~\cite{burgarth2014identifiability}, protein sequencing~\cite{acharya2010reconstructing,acharya2014quadratic,acharya2015string}, error-correcting codes for polymer-based data storage~\cite{gabrys2020mass} and  DNA mapping~\cite{karp1993mapping,skiena1994partial}. In the latter, the turnpike problem is known in the literature as the partial digest problem.
The turnpike problem is equivalent to 
the problem of recovering a sparse signal from its aperiodic autocorrelation, which has been studied in depth in the phase retrieval literature (but is different from the crystallographic phase retrieval problem)~\cite{ranieri2013phase,jaganathan2016phase,bendory2017fourier}. 
In~\cite{skiena1990reconstructing}, the maximum possible number $H(k)$ of  non-equivalent and homometric sets of $k$ points in $\mathbb{R}$ was bounded 
by $ \frac{1}{2}k^{0.8107144} \leq H(k) \leq \frac{1}{2}k^{1.2324827}$. 
If the support is collision-free, then the difference set always determines the points (up to a shift and reflection), with the exception of when the points $t_i$ belong to an explicit $2$-dimensional subspace of points when $k=6$~\cite{bekir2007there,bloom1977counterexample,bloom1977applications,piccard1939ensembles,ranieri2013phase}. 

If $M$ is at least the maximum distance between the points
in $S$, then the set $S$ can be embedded in the half-circle by
the map $a_i \mapsto t_i = (\cos (\pi a_i/M), \sin(\pi a_i/M))$.
This follows since the $t_i$'s lie in a half circle $0 \leq \pi(|a_i - a_j|/M) \leq \pi$
and the 
inner product $t_i \cdot t_j= \cos(\pi(a_i-a_j)/M)$ determines
$|a_i -a_j|$. Thus, the problem of recovering $S = \{t_1, \ldots ,t_k\}$
from its pairwise distances 
is equivalent to the problem of recovering the $\O(2)$-orbit of the binary $\delta$-function
$x = \sum_{i=1}^k \delta_{t_i}$ from its second moment.

Note that the set $S= \{a_1, \ldots , a_k\}\subset \R$ is collision-free
if and only if the corresponding set $T = \{t_1, \ldots ,t_k\} \subset S^1$ is collision-free. The result of~\cite{bekir2007there} implies
that the only collision-free subsets of the real line that cannot
be determined, up to a rigid motion, from their pairwise differences 
are six element sets of the form 
$P = \{0, a, b - 2a, 2b- 2a, 2b, 3b -a\}$
or
$Q = \{0, a, 2a + b, a + 2b, 2b - a, 3b - a\},$
where $a,b$ are real numbers. In this case, the sets $P,Q$ have the same
difference sets but are not equivalent under rigid transformations.
Translated to $S^1$, this implies that the only $\O(2)$-orbits of binary $\delta$-functions with collision-free supports lying in a half-circle which cannot be determined from their second moments are the pairs of 
$\delta$ functions $x = \sum_{t_i \in T} \delta_{t_i}$ and 
$y = \sum_{s_i \in S} \delta_{s_i}$, where 
 $S = \{( \cos(t/2M),\sin(t/2M))\}_{t \in P}$
and $T = \{ (\cos (t/2M), \sin(t/2M))\}_{t \in Q}$.

\begin{example} \label{ex.piccard}
The sets of integer points $P= \{0, 1, 8,11,13,17\}$ and $Q=\{0,1,4,10,12,17\}$
are collision-free and have the same difference sets, but are not 
equivalent. If we embed these points in the half-circle
by the map $\Z \to S^1$, $n \mapsto (\cos(\pi n/17), \sin(\pi n/17))$, 
then we obtain two sets of points $S = \{s_1, \ldots , s_6\}$
and $T = \{t_1, \ldots , t_6\}$ such that $x= \sum_{i=1}^6 \delta_{t_i}$
and $y = \sum_{i=1}^6 \delta_{s_i}$ are not orthogonally
equivalent but have the same second moments. With this ordering of the points, 
the corresponding Gram matrices are 
{\setlength{\arraycolsep}{3.6pt}
{\scriptsize{
\begin{align*}
   A = \left(
\begin{array}{cccccc} 
 1.0 & 0.98 & 0.74 & -0.27 & -0.60 & -1.0 \\
 0.98 & 1. & 0.85 & -0.092 & -0.45 & -0.98 \\
 0.74 & 0.85 & 1. & 0.45 & 0.092 & -0.74 \\
 -0.27 & -0.092 & 0.45 & 1. & 0.93 & 0.27 \\
 -0.60 & -0.45 & 0.092 & 0.93 & 1. & 0.60 \\
 -1.0 & -0.98 & -0.74 & 0.27 & 0.60 & 1.0 \\
\end{array}
\right),  \quad 
B =\left(
\begin{array}{cccccc}
 1.0 & 0.98 & 0.092 & -0.45 & -0.74 & -1.0 \\
 0.98 & 1. & 0.27 & -0.27 & -0.60 & -0.98 \\
 0.092 & 0.27 & 1. & 0.85 & 0.60 & -0.092 \\
 -0.45 & -0.27 & 0.85 & 1. & 0.93 & 0.45 \\
 -0.74 & -0.60 & 0.60 & 0.93 & 1. & 0.74 \\
 -1.0 & -0.98 & -0.092 & 0.45 & 0.74 & 1.0 \\
\end{array}
\right)
\end{align*} 
}}}
which have the same set of entries, but one cannot be obtained
from the other by the action of the permutation group $S_6$ which simultaneously permutes rows and columns.

By Theorem \ref{thm.secondmoment}, there are at most $14!/6!\approx 1.21\times 10^8$ possible $\delta$-functions with the same second moment. However, in the example with six points on $S^1$ we actually only obtain two such functions. The reason is that in this case, there is only one non-trivial permutation of the Gram matrix $X^TX$ that can be factored as $Y^TY$,  
where $Y$ is a $2 \times 6$ matrix. Indeed, a general permutation
of $X^TX$ need not be semi-definite nor have rank 2.
For example, the matrix
{\scriptsize{
$$C = \left(
\begin{array}{cccccc}
 1.0 & 0.74 & 0.98 & -0.27 & -0.60 & -1.0 \\
 0.74 & 1. & 0.85 & -0.092 & -0.45 & -0.98 \\
 0.98 & 0.85 & 1. & 0.45 & 0.092 & -0.74 \\
 -0.27 & -0.092 & 0.45 & 1. & 0.93 & 0.27 \\
 -0.60 & -0.45 & 0.092 & 0.93 & 1. & 0.60 \\
 -1.0 & -0.98 & -0.74 & 0.27 & 0.60 & 1.0 \\
\end{array}
\right)$$ }}
has eigenvalues $\{3.9,2.1,0.30,-0.28,0,0\}$, so it is not positive semi-definite and has rank 4.
\end{example}

\begin{remark} The result of \cite{bekir2007there} is only relevant for
binary $\delta$-functions with collision 
free support lying in a half circle.
We expect that there are other examples of binary $\delta$-functions with collision-free
support on $S^1$ which cannot be recovered from their second moments. 
\end{remark}

\section*{Acknowledgments}
The authors are grateful to Tanya Christiansen for helpful discussions, as well as to the referee for finding an error in an earlier version of the manuscript. T.B.  and D.E. are supported by the BSF grant no. 2020159. 
T.B. is also supported in part by the NSF-BSF grant no. 2019752, and the ISF grant no.1924/21.
D.E. was also supported by NSF-DMS 2205626.
O.M. was supported by the ARO W911NF-17-1-0512 and
the Simons Foundation Math + X Investigator award 506976.


\begin{thebibliography}{10}
\expandafter\ifx\csname url\endcsname\relax
  \def\url#1{\texttt{#1}}\fi
\expandafter\ifx\csname urlprefix\endcsname\relax\def\urlprefix{URL }\fi
\expandafter\ifx\csname href\endcsname\relax
  \def\href#1#2{#2} \def\path#1{#1}\fi

\bibitem{skiena1990reconstructing}
S.~S. Skiena, W.~D. Smith, P.~Lemke, Reconstructing sets from interpoint
  distances, in: Proceedings of the sixth annual symposium on Computational
  geometry, 1990, pp. 332--339.

\bibitem{lemke2003reconstructing}
P.~Lemke, S.~S. Skiena, W.~D. Smith, Reconstructing sets from interpoint
  distances, Discrete and Computational Geometry: The Goodman-Pollack
  Festschrift (2003) 597--631.

\bibitem{bendory2022algebraic}
T.~Bendory, D.~Edidin, Algebraic theory of phase retrieval, Not. AMS 69~(9)
  (2022) 1487--1495.

\bibitem{elser2017complexity}
V.~Elser, The complexity of bit retrieval, IEEE Transactions on Information
  Theory 64~(1) (2017) 412--428.

\bibitem{elser2018benchmark}
V.~Elser, T.-Y. Lan, T.~Bendory, Benchmark problems for phase retrieval, SIAM
  Journal on Imaging Sciences 11~(4) (2018) 2429--2455.

\bibitem{bendory2023finite}
T.~Bendory, D.~Edidin, I.~Gonzalez, Finite alphabet phase retrieval, Applied
  and Computational Harmonic Analysis 66 (2023) 151--160.

\bibitem{bendory2020toward}
T.~Bendory, D.~Edidin, Toward a mathematical theory of the crystallographic
  phase retrieval problem, SIAM Journal on Mathematics of Data Science 2~(3)
  (2020) 809--839.

\bibitem{edidin2023generic}
D.~Edidin, A.~Suresh, The generic crystallographic phase retrieval problem,
  arXiv preprint arXiv:2307.06835 (2023).

\bibitem{ghosh2022sparse}
S.~Ghosh, P.~Rigollet, Sparse multi-reference alignment: Phase retrieval,
  uniform uncertainty principles and the beltway problem, Foundations of
  Computational Mathematics (2022) 1--48.

\bibitem{duxbury2022unassigned}
P.~Duxbury, C.~Lavor, L.~Liberti, L.~L. de~Salles-Neto, Unassigned distance
  geometry and molecular conformation problems, Journal of Global Optimization
  (2022) 1--10.

\bibitem{disser2017geometric}
Y.~Disser, S.~S. Skiena, Geometric reconstruction problems, in: Handbook of
  Discrete and Computational Geometry, Chapman and Hall/CRC, 2017, pp.
  897--913.

\bibitem{huang2021reconstructing}
S.~Huang, I.~Dokmani{\'c}, Reconstructing point sets from distance
  distributions, IEEE Transactions on Signal Processing 69 (2021) 1811--1827.

\bibitem{bendory2020single}
T.~Bendory, A.~Bartesaghi, A.~Singer, Single-particle cryo-electron microscopy:
  Mathematical theory, computational challenges, and opportunities, IEEE signal
  processing magazine 37~(2) (2020) 58--76.

\bibitem{bendory2023autocorrelation}
T.~Bendory, Y.~Khoo, J.~Kileel, O.~Mickelin, A.~Singer, Autocorrelation
  analysis for cryo-{EM} with sparsity constraints: Improved sample complexity
  and projection-based algorithms, Proceedings of the National Academy of
  Sciences 120~(18) (2023) e2216507120.

\bibitem{kam1980reconstruction}
Z.~Kam, The reconstruction of structure from electron micrographs of randomly
  oriented particles, Journal of Theoretical Biology 82~(1) (1980) 15--39.

\bibitem{bendory2022sample}
T.~Bendory, D.~Edidin, The sample complexity of sparse multireference alignment
  and single-particle cryo-electron microscopy, SIAM Journal on Mathematics of
  Data Science 6~(2) (2024) 254--282.

\bibitem{abbe2018estimation}
E.~Abbe, J.~M. Pereira, A.~Singer, Estimation in the group action channel, in:
  2018 IEEE International Symposium on Information Theory (ISIT), IEEE, 2018,
  pp. 561--565.

\bibitem{perry2019sample}
A.~Perry, J.~Weed, A.~S. Bandeira, P.~Rigollet, A.~Singer, The sample
  complexity of multireference alignment, SIAM Journal on Mathematics of Data
  Science 1~(3) (2019) 497--517.

\bibitem{bendory2017bispectrum}
T.~Bendory, N.~Boumal, C.~Ma, Z.~Zhao, A.~Singer, Bispectrum inversion with
  application to multireference alignment, IEEE Transactions on signal
  processing 66~(4) (2017) 1037--1050.

\bibitem{bandeira2023estimation}
A.~S. Bandeira, B.~Blum-Smith, J.~Kileel, J.~Niles-Weed, A.~Perry, A.~S. Wein,
  Estimation under group actions: recovering orbits from invariants, Applied
  and Computational Harmonic Analysis 66 (2023) 236--319.

\bibitem{bendory2023phase}
T.~Bendory, N.~Dym, D.~Edidin, A.~Suresh, Phase retrieval with semi-algebraic
  and {ReLU} neural network priors, arXiv preprint arXiv:2311.08833 (2023).

\bibitem{bendory2024transversality}
T.~Bendory, N.~Dym, D.~Edidin, A.~Suresh, A transversality theorem for
  semi-algebraic sets with application to signal recovery from the second
  moment and cryo-{EM}, arXiv preprint arXiv:2405.04354 (2024).

\bibitem{sharon2020method}
N.~Sharon, J.~Kileel, Y.~Khoo, B.~Landa, A.~Singer, Method of moments for {3D}
  single particle ab initio modeling with non-uniform distribution of viewing
  angles, Inverse Problems 36~(4) (2020) 044003.

\bibitem{levin20183d}
E.~Levin, T.~Bendory, N.~Boumal, J.~Kileel, A.~Singer, 3d ab initio modeling in
  cryo-{EM} by autocorrelation analysis, in: 2018 IEEE 15th International
  Symposium on Biomedical Imaging (ISBI 2018), IEEE, 2018, pp. 1569--1573.

\bibitem{boutin2007point}
M.~Boutin, G.~Kemper, Which point configurations are determined by the
  distribution of their pairwise distances?, International Journal of
  Computational Geometry \& Applications 17~(01) (2007) 31--43.

\bibitem{karanam2018magnitude}
C.~R. Karanam, B.~Korany, Y.~Mostofi, Magnitude-based angle-of-arrival
  estimation, localization, and target tracking, in: 2018 17th ACM/IEEE
  International Conference on Information Processing in Sensor Networks (IPSN),
  IEEE, 2018, pp. 254--265.

\bibitem{burgarth2014identifiability}
D.~Burgarth, K.~Yuasa, Identifiability of open quantum systems, Physical Review
  A 89~(3) (2014) 030302.

\bibitem{acharya2010reconstructing}
J.~Acharya, H.~Das, O.~Milenkovic, A.~Orlitsky, S.~Pan, On reconstructing a
  string from its substring compositions, in: 2010 IEEE International Symposium
  on Information Theory, IEEE, 2010, pp. 1238--1242.

\bibitem{acharya2014quadratic}
J.~Acharya, H.~Das, O.~Milenkovic, A.~Orlitsky, S.~Pan, Quadratic-backtracking
  algorithm for string reconstruction from substring compositions, in: 2014
  IEEE International Symposium on Information Theory, IEEE, 2014, pp.
  1296--1300.

\bibitem{acharya2015string}
J.~Acharya, H.~Das, O.~Milenkovic, A.~Orlitsky, S.~Pan, String reconstruction
  from substring compositions, SIAM Journal on Discrete Mathematics 29~(3)
  (2015) 1340--1371.

\bibitem{gabrys2020mass}
R.~Gabrys, S.~Pattabiraman, O.~Milenkovic, Mass error-correction codes for
  polymer-based data storage, in: 2020 IEEE International Symposium on
  Information Theory (ISIT), IEEE, 2020, pp. 25--30.

\bibitem{karp1993mapping}
R.~M. Karp, Mapping the genome: some combinatorial problems arising in
  molecular biology, in: Proceedings of the twenty-fifth annual ACM symposium
  on Theory of computing, 1993, pp. 278--285.

\bibitem{skiena1994partial}
S.~S. Skiena, G.~Sundaram, A partial digest approach to restriction site
  mapping, Bulletin of Mathematical Biology 56 (1994) 275--294.

\bibitem{ranieri2013phase}
J.~Ranieri, A.~Chebira, Y.~M. Lu, M.~Vetterli, Phase retrieval for sparse
  signals: Uniqueness conditions, arXiv preprint arXiv:1308.3058 (2013).

\bibitem{jaganathan2016phase}
K.~Jaganathan, Y.~C. Eldar, B.~Hassibi, Phase retrieval: An overview of recent
  developments, Optical Compressive Imaging (2016) 279--312.

\bibitem{bendory2017fourier}
T.~Bendory, R.~Beinert, Y.~C. Eldar, Fourier phase retrieval: Uniqueness and
  algorithms, in: Compressed Sensing and its Applications: Second International
  MATHEON Conference 2015, Springer, 2017, pp. 55--91.

\bibitem{bekir2007there}
A.~Bekir, S.~W. Golomb, There are no further counterexamples to {S}.
  {P}iccard's theorem, IEEE transactions on information theory 53~(8) (2007)
  2864--2867.

\bibitem{bloom1977counterexample}
G.~S. Bloom, A counterexample to a theorem of {S}. {P}iccard, Journal of
  Combinatorial Theory, Series A 22~(3) (1977) 378--379.

\bibitem{bloom1977applications}
G.~S. Bloom, S.~W. Golomb, Applications of numbered undirected graphs,
  Proceedings of the IEEE 65~(4) (1977) 562--570.

\bibitem{piccard1939ensembles}
S.~Piccard, Sur les ensembles de distances des emsembles de points d'un espace
  euclidien, Neuch\^{a}tel: Secr\'{e}tariat de l'Universit\'{e}, 1939.

\end{thebibliography}

\bibliographystyle{elsarticle-num}

\end{document}